\documentclass[a4paper]{article}

\usepackage{layaureo}
\usepackage{natbib}
\usepackage{amsmath}
\usepackage{amssymb}
\usepackage{amsthm}
\usepackage{thmtools}  
\usepackage{siunitx}
\usepackage{hyperref}
\usepackage[capitalise]{cleveref}
\usepackage{bm}
\usepackage{tikz}
\usepackage{pgfplotstable}
\usepackage{pgfplots}
\pgfplotsset{compat=1.18}

\usetikzlibrary{arrows.meta}

\theoremstyle{plain}
\newtheorem{theorem}{Theorem}[section]

\newtheorem{proposition}[theorem]{Proposition}

\theoremstyle{definition}
\newtheorem{definition}[theorem]{Definition}
\newtheorem{example}[theorem]{Example}
\newtheorem{hypothesis}[theorem]{Hypothesis}

\theoremstyle{remark}
\newtheorem{remark}[theorem]{Remark}

\numberwithin{equation}{section}

\newcommand{\quoted}[1]{``#1''} 

\newcommand{\vc}[1]{{\bm #1}} 
\def\eps{\varepsilon}

\DeclareMathOperator*{\EV}{\mathbb{E}}
\DeclareMathOperator*{\esssup}{ess\,sup}

\newcommand{\Aspace}{\mathcal{A}}
\newcommand{\Sspace}{\mathcal{S}}

\newcommand{\R}{\mathbb{R}}
\newcommand{\N}{\mathbb{N}}
\DeclareMathOperator*{\EVp}{\mathbb{E}_\pi}
\newcommand{\Var}{\mathbb{V}\mathrm{ar}}

\title{Time-inhomogeneous volatility aversion for financial applications of reinforcement learning}

\author{Federico Cacciamani\thanks{Cross Asset Systematic Trading, IMI Corporate and Investment Banking, Intesa Sanpaolo.} \and Roberto Daluiso\thanks{Interest Rates and Credit Models, IMI Corporate and Investment Banking, Intesa Sanpaolo. Corresponding author: \texttt{roberto.daluiso@intesasanpaolo.com}.} \and Marco Pinciroli\thanks{XVA Management and KVA Pricing, IMI Corporate and Investment Banking, Intesa Sanpaolo.} \and Michele Trapletti\footnotemark[3] \and Edoardo Vittori\footnotemark[1]}
\date{February 5, 2026}
\providecommand{\keywords}[1]{\small\textbf{Keywords:} #1}

\begin{document}

\maketitle

\begin{abstract}
	In finance, sequential decision problems are often faced, for which reinforcement learning (RL) emerges as a promising tool for optimisation without the need of analytical tractability. However, the objective of classical RL is the expected cumulated reward, while financial applications typically require a trade-off between return and risk. In this work, we focus on settings where one cares about the time split of the total return, ruling out most risk-aware generalisations of RL which optimise a risk measure defined on the latter. We notice that a preference for homogeneous splits, which we found satisfactory for hedging, can be unfit for other problems, and therefore propose a new risk metric which still penalises uncertainty of the single rewards, but allows for an arbitrary planning of their target levels. We study the properties of the resulting objective and the generalisation of learning algorithms to optimise it. Finally, we show numerical results on toy examples.
\end{abstract}


\section{Introduction}\label{sec:introduction}

In many sequential decision problems, maximisation of expected reward is not the only objective, as risk must be kept under control.
Once one has defined a way to quantify this risk, then by elementary Lagrange multiplier theory, without loss of generality one can maximise a single objective where the risk metric multiplied by a risk aversion coefficient is subtracted from the expected reward.
However, the risk metric should be carefully crafted to both accurately describe the goal and allow efficient optimisation.
The former requirement is in general application-specific.

In particular, classical risk measures defined on the total return do not completely align with the needs of trading applications, as they are insensitive to the time distribution of profits and losses, which financial institutions care about for reasons related to financing, budgeting, etc.
To overcome this issue, \citet{Bisi2020trvo} proposed to measure risk as a variance across both time and population (which they call \quoted{reward volatility}), thus expressing a preference not only for low randomness of the return, but also for a homogeneous split thereof over time steps.

This preference for homogeneity has proven practically appropriate for hedging applications, see e.g.~\citet{vittori2020option, mandelli2023hedging, daluiso2023Acva}; although the extended version of the latter \citep{daluiso2023Bcva} already underlines that when the horizon is random, one must carefully define in which sense the dispersion of rewards should be homogeneous across realisations.
However, in fixed-horizon problems, which are \textit{a fortiori} non stationary, forcing comparable returns at different times may introduce an undesired distortion into the objective function.\footnote{Note that the hedging problems in the above papers are themselves formulated with finite horizons; in those specific applications, good policies are nonetheless expected to produce uniform reward histories, so that volatility-penalised optimisation is appropriate, but in a sense, this is an accident.}
For instance, in so-called optimal execution tasks which consist in the liquidation of a large amount of an asset, when marking-to-market the value of the residual position, the impact of transactions on prices can easily introduce a larger \emph{expected} loss at the beginning, when most of the asset to be liquidated is still in the portfolio: plain reward volatility optimisation would improperly strive to counterbalance this natural pattern, and we believe one should not (see \cref{sec:execution} for details).

In this paper, we design a different family of risk-aware objective functions, which retains a tractable policy gradient to be used in advanced policy search or actor-critic algorithms, and still penalises risk (in the sense of probabilistic uncertainty) of the \emph{single} rewards, but allows for an arbitrary planning of their \emph{individual} target values. The rationale is that a completely deterministic reward path known in advance should not be considered \quoted{risky} in any financially meaningful sense, whatever its time profile.

After a review of relevant literature in \cref{sec:literature}, the body of the paper is organised in two main sections.
In \cref{sec:theory} we define the proposed objective functions and study their properties, introducing reinforcement learning based algorithms to optimise them.
In \cref{sec:examples} we observe the behaviour of the optimal policies on three example problems: the first one is an artificial example solvable in closed form which should clarify the difference between our concept of risk aversion and existing alternatives; the second one is a financial example showing the efficacy of our setup on a problem with a fixed deterministic horizon; while the third one is a classical non financial reinforcement learning test case, where some unanticipated effects of stochastic horizons can be appreciated.
\Cref{sec:conclusion} concludes the paper with some final considerations.

\subsection{Related literature}\label{sec:literature}

By the present day, there is a large literature on risk aversion in the machine learning subfield dedicated to sequential decision problems, known as Reinforcement Learning (RL): we refer to \cite{garcia15a} for a comprehensive review. For our purpose, we only attempt a broad conceptual categorisation of the optimisation criteria and list only a few notable papers for each category, just as examples and starting points for the interested reader.

The simplest approach is to apply a non-linear transformation to each realised reward \citep[see e.g.][]{shen2014risktrading}; Conceptually, such a path-by-path transformation represents more a subjective weighting of the objective outcomes (e.g., a different relevance given to monetary profits and losses), than a measure of \quoted{risk} in the proper financial meaning of the term, namely a preference ordering on the probability distribution of rewards.

A more principled approach is to apply a non-linear risk measure to the random variable representing the sum of discounted rewards \citep[see e.g.][]{moody2001learning, Tamar_variance, morimura2010nonparametric, tamar_variance_2013, prashanth_actor-critic_2014, tamar2015policy, tamar_sequential_2017, chow2017risk}. In such setting, unless one finds a clever algorithm crafted for a very specific choice of risk measure, generally speaking the problem becomes very unsuitable for genuine RL approaches, as one is forced to attribute the full result to the last timestep of the episode, with no intermediate rewards. For instance, \cite{neagu2025deep} tries different state-of-the-art algorithms for hedging in this formulation, and empirically finds satisfactory results only with direct policy search (in a way, more a supervised than a reinforcement learning algorithm). However, gradient descent on policy parameters is only possible if the full environment is a differentiable function of actions, which is often not the case in other applications, or daunting to implement even when theoretically true.

Approaches where the time split of rewards matters are less common: besides the already discussed \citet{Bisi2020trvo}, we mention the proposals based on dynamic risk measures \citep{coache2023reinforcement, marzban2023deep}, which are designed to have time-consistent optimal policies, i.e.~optimal also for the conditional risk measure as seen from any future state which can be reached. This means that these approaches take care of the risk of \emph{co-terminal} problems, which is complementary to the view of the present paper where our focus on planning/budgeting could in some sense be understood more in terms of \emph{co-starting} risk. This also explains why we are not particularly interested in consistency properties, as obviously in such settings the starting time is righteously special as the occasion when agents lay down their plans. See also \cref{sec:toy} for a comparison on a toy example.

To conclude on risk-aware generalisations of RL beyond the scope of finance, it is worth mentioning the growing literature in robust reinforcement learning: see \citet{moos2022robust} for a review. Here \quoted{risk} takes yet another meaning related to prudent optimisation under a \emph{set} of candidate probability measures, which in financial jargon would map to the concepts of \quoted{model risk} and \quoted{Knightian uncertainty}, while here we hope to manage \emph{market} risk.

Turning to finance-specific contributions, one of the richest streams is related to the problem of dynamic hedging of a derivative portfolio. Within such field often named \quoted{deep hedging}, some more works looking into stepwise rewards represented by the period Profit and Loss (PnL) can be found, as doing so is clearly more aligned with real-life practice. Such papers can usually introduce risk in interesting but \emph{ad hoc} ways given their scoping to this specific application. In particular, \citet{kolm2019dynamic, du2020deep} justify a reward which is a linear-quadratic transformation of the PnL by heuristically first proxying the mean-variance with the sum of the mean-variances of the rewards, and then proxying the variance of the reward by its expected square, guessing that the square of the mean period PnL should be small. \citet{halperin2019qlbs,halperin2020qlbs} instead penalises the conditional co-terminal variance of the value of the hedging portfolio, arguing that the value of the derivative should be defined by the hedging process itself. Finally, among the deep hedging papers introducing more traditionally a risk measure on the final return, we just mention \citet{mueller2024fast} because they also feel the need to cleanup the penalised variance from a component of reward which is not a direct measure of residual market risk, namely transaction costs; however, this modification is very specific to the application, and would not coincide with our definition even on a single-step problem, as costs introduce genuine variability also at each fixed time due to the randomness of the rebalancing amounts, whose risk is still penalised in our objective functions.

Finally, as already anticipated, one of the key ingredients in our definitions is the tuning of deterministic targets to the single rewards, and a similar idea of setting a unique target for the full return can be recognised in previous works, including at least the seminal \citet{buehler2019deep} and the recent \citet{han2025risksensitive}.

\section{Theory}\label{sec:theory}

\subsection{Definitions and properties}\label{sec:definitions}

This section presents our proposed risk-aware reinforcement learning formulation, by first proposing a modification of the mean-volatility criterion, and then recognizing that it can be seen as an instance of a larger family of objective functions.

We work in the context of an inhomogeneous Markov Decision Process and intentionally generalise the homogeneous notation of \cite{daluiso2023Bcva}, as the definitions and proofs will be similar. In particular, we denote by $(s_i)_{i \geq 0}$ the sequence of states in the state space $\Sspace$; by $(a_i)_{i \geq 0}$ the sequence of actions in the action space $\Aspace$; by $\mu$ the initial state distribution; by $\eps \in (0,\infty]$ the number of steps; by $\gamma$ the discount factor, and by $\Gamma$ their cumulated value $\sum_{i=1}^{\eps} \gamma^{i-1}$; by $(r_i)_{i > 0}$ the sequence of real-valued rewards, and by $\mathcal{G}$ the return $\sum_{i=1}^{\eps} \gamma^{i-1} r_i$. We assume without loss of generality that the Markov chain $s$ is extended to time steps beyond $\eps$ in a way that $\eps < i$ is a function of $s_i$, e.g.~by introduction of an absorbing post-episode-end state.

On the other hand, to account for time inhomogeneity, the law of $s_{i+1}$ given the state $s_i$ and the action $a_i$ at the $i$-th time step will be indexed by $i$, using the notation $\mathcal{P}_i(\cdot|s_i,a_i)$. Similarly, the model for $r_{i+1}$ will be described by a conditional law $\mathcal{R}_i(\cdot|s_i,a_i,s_{i+1})$, and the randomised policy $\pi$ by the conditional laws $\pi_i(\cdot|s_i)$ of $a_i$ given the state $s_i$ at the $i$-th step. 

We now introduce some new definitions. The expected value of the $i$-th reward is
\begin{equation}
	\bar{r}_{\pi,i} := \EVp [ r_i | \eps \geq i ],
\end{equation}
and we define a risk metric which penalises deviation from this value:

\begin{definition}[Inhomogeneous reward volatility]\label{def:inhomogeneous_vol}
	The inhomogeneous reward volatility or inhomogeneous reward variance is:
	\begin{align}
		\varsigma^2_\pi
		 & = \EVp \left[\sum_{i=1}^{\eps} \gamma^{i-1} \left(r_i-\bar{r}_{\pi,i}\right)^2\right]
		\label{eq:inhomogeneous_variance}
	\end{align}
\end{definition}

\begin{definition}[Inhomogeneous mean-volatility]\label{def:inhomogeneous_meanvol} Given a risk aversion coefficient ${\beta\geq 0}$, the ranking criterion for policies is the inhomogeneous mean-volatility functional
	\begin{equation}\label{eq:inhomogeneous_meanvar}
		\pi \rightarrow \EVp\left[\mathcal{G}\right] - \beta \varsigma^2_\pi.
	\end{equation}
\end{definition}

\begin{remark}[Comparison to homogeneous definitions]
	The above differ from unnormalised reward volatility and mean-volatility as defined in \citet{daluiso2023Bcva}
	\[
		\hat{\nu}^2_\pi = \EVp \left[\sum_{i=1}^{\eps} \gamma^{i-1} \left(r_i-J_{\pi}\right)^2\right] \text{ and }\EVp\left[\mathcal{G}\right] - \beta \hat{\nu}^2_\pi, \quad \text{ where } J_{\pi} := \EVp\left[\Gamma^{-1} \mathcal{G}\right]
	\]
	precisely in that the latter subtract a common average value $J_\pi$ from each reward $r_i$, while in \eqref{eq:inhomogeneous_variance} we have a different central value $\bar{r}_{\pi,i}$ for each time step $i$.
\end{remark}

\begin{remark}[Irreducibility to classical RL]
	The simpler way to introduce risk aversion in a sequential decision problem is to apply a nonlinear transformation to the reward profile and then use a classical RL formulation. Sometimes this is even conceptually equivalent to more complex formulations, though maybe less efficient algorithmically. For instance, a Pareto-optimal mean-variance policy maximises variance among all policies with given expected return, and hence also maximises $\EVp[\mathcal{G}^2]$ for fixed $\EVp[\mathcal{G}]$: therefore, to determine the Pareto frontier, instead of using an \textit{ad hoc} algorithm to optimise $\EVp[\mathcal{G}] - \beta\,\Var_\pi[\mathcal{G}]$, one may in principle optimise classically the expected value of a modified terminal payoff $\mathcal{G} - \beta' \mathcal{G}^2$, where we introduced a new symbol $\beta'$ to clarify that the solution to the original problem for fixed $\beta$ will be found for a different and \textit{a priori} unknown value $\beta'$ of the Lagrange multiplier. Analogously, with homogeneous mean-volatility in fixed horizon problems where $\eps$ is almost surely constant, the risk criterion is
	\[
		\hat{\nu}_\pi^2 = \sum_{i=1}^{\eps} \gamma^{i-1} \EVp \left[ r_i^2 \right] +  \sum_{i=1}^{\eps} \gamma^{i-1} J_\pi^2 - 2 \sum_{i=1}^{\eps} \gamma^{i-1} \EVp \left[ r_i \right] J_\pi = \EVp \left[\sum_{i=1}^{\eps} \gamma^{i-1} r_i^2 \right] - \Gamma J_\pi^2,
	\]
	where $\Gamma J_\pi = \EVp[\mathcal{G}]$ is exactly the return criterion: therefore, Pareto-optimal homogeneous mean-volatility policies should be spanned by classical optimisers of the transformed reward $r_i - \beta' r_i^2$.

	This is not the case for the new objective \eqref{eq:inhomogeneous_meanvar}. Indeed, even for constant $\eps$, while one can write the risk as $\varsigma^2_\pi = \sum_{i=1}^{\eps} \gamma^{i-1} \Var_\pi \left[ r_i \right]$,
	expanding the variances one gets
	\[
		\EVp \left[ \sum_{i=1}^{\eps} \gamma^{i-1}  r_i^2 \right] - \sum_{i=1}^{\eps} \gamma^{i-1} \bar{r}_{i,\pi}^2
	\]
	and this time the correction term $\sum_{i=1}^{\eps} \gamma^{i-1} \bar{r}_{i,\pi}^2$ to the plain return of the squared reward is not a function of the return criterion $\EVp[\mathcal{G}]$.
\end{remark}

Quantitatively, one can note that the risk definition \eqref{eq:inhomogeneous_variance} is weaker than that of reward volatility $\hat{\nu}^2_\pi$:
\begin{theorem}[Volatility inequality]
	For any policy $\pi$, it holds that
	\begin{equation}
		\hat{\nu}^2_\pi \geq \varsigma^2_\pi.
	\end{equation}
\end{theorem}
\begin{proof}
	The $i$-th addend in $\hat{\nu}^2_\pi$ is
	\begin{equation}\label{eq:1}
		\EVp \left[\gamma^{i-1}(r_i - J_\pi)^2I_{\eps\geq i}\right] = \gamma^{i-1}\EVp \left[ \EVp\left[(r_i - J_\pi)^2 | \eps \geq i \right] I_{\eps\geq i}\right];
	\end{equation}
	by applying to the law $\EVp\left[\cdot | \eps \geq i \right]$ the elementary property that the constant $x$ which minimises the expectation of the quadratic error $(r_i - x)^2$ is the mean, we get that
	\[
		\EVp\left[(r_i - J_\pi)^2 | \eps \geq i \right] \geq \EVp\left[(r_i - \EVp\left[r_i | \eps \geq i \right])^2 | \eps \geq i \right] = \EVp\left[(r_i - \bar{r}_{\pi,i})^2 | \eps \geq i \right].
	\]
	Substituting back in \eqref{eq:1}, we get that the generic addend of $\hat{\nu}^2_\pi$ is larger than that of $\varsigma^2_\pi$, concluding the proof.
\end{proof}

The comparison with the standard return variance is less neat, as no inequality holds in general if $\eps$ is stochastic. Indeed, a sequence of deterministic rewards stopped at a random episode end step has zero inhomogeneous reward volatility but positive return variance. On the other hand, if the horizon is deterministic as in most problems, then the new risk measure is stricter than the return variance $\sigma_\pi^2 = \Var_\pi \left[ \mathcal{G} \right]$:

\begin{theorem}[Variance inequality]
	For any policy $\pi$, if $\eps$ is deterministic (finite or infinite), it holds that
	\begin{equation}
		\sigma^2_\pi \leq \Gamma\varsigma^2_\pi.
	\end{equation}
\end{theorem}
\begin{proof}
	For deterministic $\eps$, the expected return appearing in the definition of variance is equal to the sum of $\gamma^{i-1}\bar{r}_{\pi,i}$. Then one can apply to the inner summation in the definition of $\sigma_\pi^2$ a discrete Jensen inequality:
	\begin{align*}
		\sigma_\pi^2 & = \EVp_{s_0 \sim \mu}\left[\Gamma^2\left(\sum_{i=1}^{\eps}\Gamma^{-1}\gamma^{i-1}(r_i - \bar{r}_{\pi,i})\right)^2 \right]                            \\
		             & \leq \EVp_{s_0 \sim \mu}\left[\Gamma^2\sum_{i=1}^{\eps}\Gamma^{-1}\gamma^{i-1}\left(r_i - \bar{r}_{\pi,i}\right)^2 \right] = \Gamma \varsigma^2_\pi.
	\end{align*}
\end{proof}

The inhomogeneous mean-volatility criterion in \eqref{eq:inhomogeneous_meanvar} solves the time-attribution issues of mean-variance and homogeneous mean-volatility, but it retains other drawbacks of using a variance as a risk criterion, which essentially derive from penalizing in the same way positive and negative deviations from the expected reward $\bar{r}_{\pi,i}$. In contexts where this is not acceptable, one can use the following generalisation.

\begin{definition}[Inhomogeneous $\ell$-volatility]\label{def:inhomogeneous_ellvol}
	For any function $\ell: \R \to \R^+$ such that $\ell(0) = 0$ and any policy $\pi$, the inhomogeneous $\ell$-volatility or inhomogeneous $\ell$-variance is:
	\begin{equation}
		\varsigma^\ell_\pi
		= \inf_{\bar{r}^{\ell} \in \R^{\N}} \EVp \left[\sum_{i=1}^{\eps} \gamma^{i-1} \ell\left(r_i-\bar{r}^{\ell}_i\right)\right]
		\label{eq:inhomogeneous_ell_variance}
	\end{equation}
	and the inhomogeneous mean-$\ell$-volatility is:
	\begin{equation}
		\EVp[\mathcal{G}] - \varsigma^\ell_\pi.\label{eq:inhomogeneous_ell_meanvar}
	\end{equation}
\end{definition}

One immediately sees that \eqref{eq:inhomogeneous_ell_variance}-\eqref{eq:inhomogeneous_ell_meanvar} reduces to \eqref{eq:inhomogeneous_variance}-\eqref{eq:inhomogeneous_meanvar} for $\ell = \beta(\cdot)^2$, since in such case the optimal target reward $\bar{r}^{\ell}_{\pi,i}$ is the conditional mean $\bar{r}_{\pi,i}$; while for instance $\ell(x) = \beta| x |$ gives a risk penalty equal to a risk aversion coefficient $\beta$ times the absolute distance of the reward from the conditional median of the reward distribution. On the other hand, the generalised mean-$\ell$-volatility can be rewritten as a superposition of conditional performance metrics:
\begin{align}
	\EVp [\mathcal{G}] - \varsigma^\ell_\pi & = \sum_{i=1}^{\infty} \gamma^{i-1} \mathbb{P}(\eps \geq i) \sup_{\bar{r}^{\ell}_i \in \R}\EVp \left[r_i - \ell\left(r_i-\bar{r}^{\ell}_i\right) \bigg| \eps \geq i \right] \nonumber                          \\
	                                        & = \sum_{i=1}^{\infty} \gamma^{i-1} \mathbb{P}(\eps \geq i) \sup_{\bar{r}^{\ell}_i \in \R}\EVp \left[\bar{r}^{\ell}_i + u\left(r_i-\bar{r}^{\ell}_i\right) \bigg| \eps \geq i \right] \label{eq:superposition}
\end{align}
where each of suprema can be interpreted as the cash-invariant hull \citep[as defined in][]{filipovic2007monotone} of an expected utility computed under a conditional law, offering a link to the literature for good choices of $u$ and hence of $\ell(x) = x - u(x)$, as in the below examples.

\begin{example}[Inhomogeneous optimised certainty equivalent]\label{ex:oce}
	For an increasing concave $u$ with $u(0) = 0$ and $u'(0) = 1$, each summand becomes an optimised certainty equivalent (OCE) as introduced by \citet{bental1986oce}; this was already used in finance for a sequential decision problem at least by the influential \citet{buehler2019deep}, but computing the OCE on the total return and therefore with a single scalar target, and then optimizing it simultaneously with the policy by direct gradient descent, which needs differentiability of the environment.
\end{example}

\begin{example}[Inhomogeneous monotone mean-volatility]\label{ex:monotone_meanvar}
	For $u$ equal to the monotone hull of quadratic utility (which also fits in the previous example), i.e.
	\[
		u_{\mathrm{m}}^\beta(x) = \min\left(x, \frac{1}{2\beta}\right) - \beta \min\left(x, \frac{1}{2\beta}\right)^2,
	\]
	by trivial generalisation of what \citet{cerny2020semimartingale} proves with $\beta=1/2$, the cash invariant hull of expected utility which appears in in each term of \eqref{eq:superposition} has a nice interpretation as the monotone hull of mean-variance, i.e.~the minimal correction to mean-variance which ensures that adding a positive random variable does not worsen the measure:
	\begin{equation}\label{eq:monotone}
		\sup_{\bar{r}_i \in \R} \EVp \left[\eta + u_{\mathrm{m}}^\beta\left(X-\eta\right) | \eps \geq i \right] = \sup_{Y \in L_0^+} \left\{ \EVp[X - Y | \eps \geq i] - \beta \Var_\pi[X - Y | \eps \geq i] \right\}.
	\end{equation}
	Note that in our parametrisation, the utility function $u_{\mathrm{m}}^\beta$ corresponds to a loss function
	\[
		\ell_{\mathrm{m}}^\beta(x) = x  - u_{\mathrm{m}}^\beta(x) = x - \min\left(x, \frac{1}{2\beta}\right) + \beta \min\left(x, \frac{1}{2\beta}\right)^2 = \begin{cases}\beta x^2& x \leq (2\beta)^{-1}\\x - (4\beta)^{-1}  & x > (2\beta)^{-1} \end{cases}
	\]
	that switches in a differentiable way from the standard quadratic penalty to a less severe affine penalty when the windfall $x$ over the target reward exceeds a threshold $(2\beta)^{-1}$.
\end{example}

For numerical purposes, the following assumption can be useful.
\begin{hypothesis}\label{hp:diff}
	$\ell$ is a differentiable function with derivative $\ell'$, and the suprema in \eqref{eq:superposition} are attained by a unique maximiser $\bar{r}^\ell_{\pi,i}$ for all $i \leq \esssup \eps$, characterised by the first order condition
	\begin{equation}\label{eq:stationarity}
		\EVp [\ell'(r_i - \bar{r}^\ell_{\pi,i}) | \eps \geq i ] = 0
	\end{equation}
\end{hypothesis}

This holds true for most interesting examples, as shown by the following.
\begin{proposition}
	Suppose that for all $i \leq \esssup \eps$, there exists some $\bar{r}_i^\ell$ such that $\EVp [\ell(r_i - \bar{r}_i^\ell) | \eps \geq i]$ is finite. Then \cref{hp:diff} holds for:
	\begin{enumerate}
		\item The inhomogeneous mean-volatility defined in \eqref{eq:inhomogeneous_meanvar};
		\item The inhomogeneous OCE defined in \cref{ex:oce} if $u \in C^1(\R)$ is strictly concave;
		\item The inhomogeneous monotone mean-volatility defined in \cref{ex:monotone_meanvar}.
	\end{enumerate}
\end{proposition}
\begin{proof}
	Point 1 is a special case of point 2 for $\ell(x) = \beta x^2$.

	For both points 2 and 3, we start noting that when $u$ is concave, $\ell$ is convex, so $\mu_i(\bar{r}_i^\ell) := \EVp [\ell(r_i - \bar{r}_i^\ell) | \eps \geq i ]$ also is, and its stationary points are global minimisers. We claim that its derivative is $-\EVp [\ell'(r_i - \bar{r}_i^\ell) | \eps \geq i ]$ at every point inside its domain: indeed, there is an $h$ such that $[\bar{r}_i^\ell - 2h, \bar{r}_i^\ell + 2h]$ is in the domain, which means that for $r \in [\bar{r}_i^\ell - h, \bar{r}_i^\ell + h]$ the derivative is uniformly bounded below and above by any fixed incremental ratios chosen inside $[\bar{r}_i^\ell - 2h, \bar{r}_i^\ell - h]$ and $[\bar{r}_i^\ell + h, \bar{r}_i^\ell + 2h]$ respectively, which are integrable by definition of domain. So stationarity is exactly \eqref{eq:stationarity}.

	We now prove that a global minimum or $\mu_i$ exists. For this, note that $\ell$ has a strict global minimum in zero. Then by convexity, $\ell$ tends to infinity for large negative and positive argument. This easily implies that if $\mu_i(\bar{r}_i^\ell) < \infty$, then $\mu_i$ is larger than $\mu_i(\bar{r}_i^\ell)$ outside a compact interval. By continuity a minimum exists in that compact subset of the domain.

	Finally, we prove uniqueness of the minimum. For point 2, this is trivial by strict convexity of $\mu_i$, which is inherited from $\ell$. For point 3, the existence of two distinct minima $\bar{r}_i, \bar{r}_i'$ would imply that $\mu_i$ is linear in the interval, and therefore that $r_i - \bar{r}_i$ is almost surely inside the interval where $\ell_{\mathrm{m}}^\beta$ is linear: i.e., almost surely greater than $(2\beta)^{-1}$. But then $\bar{r}_i'' = \bar{r}_i - (2\beta)^{-1}$ would give a strictly smaller penalty almost surely, contradicting minimality of $\bar{r}_i$.
\end{proof}

\subsection{Algorithms}\label{sec:algorithms}

One idea to optimise our risk criterion is to note that the optimisation problem can be restated exchanging the order of the suprema:
\begin{equation}\label{eq:nested}
	\sup_\pi \sup_{\bar{r}^\ell\in\R^\N} \EVp \left[\sum_{i=1}^{\eps} \gamma^{i-1}\left[r_i-\ell(r_i - \bar{r}^\ell_i)\right] \right] = \sup_{\bar{r}^\ell\in\R^{\esssup \eps}} \sup_\pi \EVp \left[\sum_{i=1}^{\eps} \gamma^{i-1}\left[r_i-\ell(r_i - \bar{r}^\ell_i)\right] \right],
\end{equation}
where the inner maximum for fixed $\bar{r}^\ell$ is a plain RL problem with a modified reward function. In simple finite horizon cases, its solution is fast enough to be called at each evaluation of an outer (global) optimisation routine over a finite-dimensional $\bar{r}^\ell$; note that if the RL algorithm is iterative, then the policy and/or value functions found for the previous candidate $\bar{r}^\ell$ should be a good starting guess for the next candidate $\bar{r}^\ell$ examined by the optimiser.

Alternatively, if all rewards can be implemented as a differentiable function of the actions, one could try direct stochastic gradient descent for \eqref{eq:nested} jointly over both a parametrisation of the policy (say the weights of a neural network) and $\bar{r}^\ell$, which can be seen as a multidimensional generalisation of the idea of \citet{buehler2019deep}.

Finally, for more complex problems, we can formulate a pathwise estimator of the policy gradient for $\varsigma^\ell_\pi$, allowing a generalisation of the Trust Region Volatility Optimisation (TRVO) algorithm which we will denote by the acronym IVO (Inhomogeneous Volatility Optimisation). To this purpose, assume that \cref{hp:diff} holds, and
define the inhomogeneous version of the action-volatility function:
\begin{equation*}
	X_{\pi,i}(s,a) := \EVp\left[\sum_{j=i+1}^{\eps} \gamma^{j-i-1} \ell(r_j - \bar{r}^\ell_{\pi,j}) \,\bigg|\, s_i=s,a_i=a\right],
\end{equation*}
which is zero if $s$ is a terminal or post-terminal state. Using the symbol ${\EV}_{\pi,i}$ to denote the conditional expectation given the path up to time $i$, we can write for it a Bellman-like equation:
\begin{equation}\label{eq:bellman}
	X_{\pi,i}(s_i,a_i) = I_{\eps > i}  {\EV}_{\pi,i}\left[\left(r_{i+1}-\bar{r}_{\pi,i+1}\right)^2\right] + \gamma {\EV}_{\pi,i} \left[X_{\pi,i+1}(s_{i+1},a_{i+1})\right].
\end{equation}

Then we have the following.
\begin{theorem}\label{thm:gradient}
	Consider policies that are absolutely continuous with respect to a reference measure $\bar{\pi}$, such that the density $\pi_{\theta,i}(\cdot|s)$ at each time step $i$ depends differentiably on real-valued parameters~$\theta$. Assume that $\eps$ is almost surely finite and that \cref{hp:diff} holds. Then under technical hypotheses (see \cref{rem:conditions}), the gradient of the inhomogeneous reward volatility satisfies
	\begin{equation}\label{eq:gradient}
		\nabla_\theta\varsigma_\pi^2 = \EVp \left[ \sum_{i=0}^{\eps-1} \gamma^{i}X_{\pi,i}(s_i,a_i)\nabla_\theta\log \pi_{\theta,i}(a_i|s_i) \right].
	\end{equation}
\end{theorem}

\begin{proof}
	We will prove by induction on $N$ the following equality:\footnote{We will follow closely the proof of the analogous claim in \citet[Theorem 3.2]{daluiso2023Bcva}, which only considers a quadratic penalty, and even for $\ell(x)=\beta x^2$ does not directly apply to our setup, as it is concerned with the time-homogeneous version of the objective function. Hence we provide all details, also for the reader's convenience and to avoid reliance on the unreviewed extended version of an applied paper.}
	\begin{multline}\label{eq:claim}
		\nabla_\theta \varsigma_\pi^2 = \EVp \left[ \sum_{i=0}^{N-1} \gamma^{i} X_{\pi,i}(s_i,a_i)\nabla_\theta \log\pi_{\theta,i}(a_i|s_i) \right.\\
			\left. - \sum_{i=1}^{N} \left(\nabla_\theta \bar{r}_{\pi,i}^\ell\right) I_{\eps \geq i} \gamma^{i-1}\ell'(r_i - \bar{r}_{\pi,i}^\ell) + \gamma^N \nabla_\theta {\EV}_{\pi,N-1} \left[ X_{\pi,N}(s_N,a_N) \right] \right].
	\end{multline}
	Such relation is enough to get the conclusion by letting $N\to\infty$, because:
	\begin{enumerate}
		\item The first summation converges to \eqref{eq:gradient}, as $X_{\pi,i}$ is zero on $\eps\leq i$.
		\item The second summation has null expected value, as $\EVp[\ell'(r_i-\bar{r}_{\pi,i}^\ell)I_{\eps \geq i}] = 0$ by conditioning on $\eps \geq i$ followed by application of \cref{hp:diff};
		\item The third addend is eventually null because it a gradient of an expectation whose integrand is identically null independently of $\theta$ on $\eps < N$, and $\eps$ is a.s.~finite.
	\end{enumerate}

	For the base $N=0$, the summations disappear and the argument of the outer expected value is the deterministic $\nabla_\theta {\EV}_{\pi} [X_{\pi,0}(s_0,a_0)]$, where the expected value equals $\varsigma_\pi^2$ by the definitions and the tower rule.

	For the induction step, we analyse the third and residual term of the claim, i.e.
	\begin{equation}\label{eq:2}
		\nabla_\theta {\EV}_{\pi,N-1} \left[ X_{\pi,N}(s_N,a_N) \right].
	\end{equation}
	We interpret the conditional expected value as a plain expectation over the conditional law of $(s_N, a_N)$. From such a point of view, we see that it depends on $\theta$ both distributionally because the conditional law depends on $\theta$, and functionally because for each $(s_N,a_N)$ fixed the integrand is a function of $\theta$ via $\pi$. Then by elementary probability theory:
	\begin{enumerate}
		\item The gradient of the distributional dependence is the expected value ${\EV}_{\pi,N-1}$ of the integrand times the conditional likelihood ratio weight:
		      \[
			      X_{\pi,N}(s_N,a_N) \nabla_\theta \log \pi_{\theta,N}(a_N|s_N).
		      \]
		\item The gradient of the functional dependence is the expected value ${\EV}_{\pi,N-1}$ of the gradient of the integrand $X_{\pi,N}$. We compute it on its Bellman expansion \eqref{eq:bellman}: the current-step term gives a gradient
		      \[
			      I_{\eps > N}{\EV}_{\pi,N}[-(\nabla_\theta\bar{r}_{\pi,N+1}^\ell)\ell'(r_{N+1}-\bar{r}_{\pi,N+1}^\ell)],
		      \]
		      while the recursive term gives
		      \[
			      \gamma \nabla_\theta {\EV}_{\pi,N}[X_{\pi,N+1}(s_{N+1},a_{N+1})].
		      \]
	\end{enumerate}

	Now note that in the induction hypothesis, \eqref{eq:2} appears within an unconditional expectation $\EV_\pi$, so when substituting points 1.~and 2.~above into \eqref{eq:claim}, we can remove in both the conditional expectations ${\EV}_{\pi,N-1}$ by the tower rule. After doing this, 1.~gives exactly the new addend in the first summation of the claim; the current-step term in 2.~gives the new addend in the second summation of the claim after the removal of ${\EV}_{\pi,N}$ by a further application of the tower rule; and the recursive term in 2.~gives the new residual exactly as it appears in the claim.
\end{proof}

\begin{remark}[Technical conditions]\label{rem:conditions}
	As in \citet[][Remark 15]{daluiso2023Bcva}, the proof relies on exchanges between expectation and differentiation / limit whose validity can be justified by appeal to standard theorems in concrete settings.
\end{remark}

\begin{remark}[Algorithmic implications]
	The main difference to the original policy gradient theorem used by TRVO is the need to estimate the target levels $\bar{r}^\ell_{\pi,j}$. For the inhomogeneous mean-volatility objective, we know that they equal expectations of the reward on scenarios where the episode is still running, and are readily estimated as sample means. For general $\ell$, they are defined implicitly as minimisers of conditional expected values, so we can estimate them by numerical minimisation of the empirical counterparts of such expected values.
\end{remark}

\section{Examples}\label{sec:examples}

\subsection{Toy example}\label{sec:toy}

In this subsection we define an \emph{ad hoc} environment highlighting the conceptual novelty of the optimality criterion proposed in this paper with respect to the main risk averse formulations in the RL literature.

Specifically, we consider a fixed-horizon problem with $\eps=2$ steps and no discounting ($\gamma=1$), where both the state space and the action space are given by the set $\mathbb{R}$ of real numbers. At $i=0$, the state equals $s_0=0$ deterministically. Then action $a_0$ determines mean and the variance of a Gaussian random variable according to which the next state is sampled; more precisely, $s_1 \sim \mathcal{N}(a_0,a_0^2)$. Finally, the state transitions to the deterministic $s_2 = 1$ irrespective of the action $a_1$, and the episode ends. The rewards are computed as state differences, i.e.~$r_1 = s_1 - s_0 = s_1$ and $r_2 = s_2 - s_1 = 1 - s_1$.

The first remark is that the return is deterministic and independent of the policy:
\[
	\mathcal{G} = s_1 + (1 - s_1) = 1.
\]
As a consequence, any policy is equally good not only according to the classical RL goal, but also for the mean-variance criterion and for all risk averse goals defined by a risk measure applied to $\mathcal{G}$.

Time-consistent dynamic risk measures as proposed by \citet{coache2023reinforcement} do not express any preference between policies either, because the objective function in such case is
\[
	\rho_{0,2} = \rho_0(r_1 + \rho_1(r_2))
\]
where $\rho_0$ and $\rho_1$ are risk measures: indeed, for our environment one can write
\[
	\rho_{0,2} = \rho_0(s_1 + \rho_1(1-s_1)) = \rho_0(s_1 + (1-s_1)) = \rho_0(1) = 1,
\]
where we used that $\rho_i(Z_i)$ is equal to $Z_i$ for any random variable known at step $i$.

On the other hand, as the expected return does not depend on the policy,  homogeneous mean-volatility would minimise the unnormalised reward volatility for any positive risk aversion coefficient, i.e.
\begin{align*}
	\hat{\nu}^2_\pi & = \EVp \left[ \left(s_1 - J_\pi\right)^2 + \left(1 - s_1 - J_\pi\right)^2 \right] = \EVp \left[ \left(s_1 - 1/2\right)^2 + \left(1 - s_1 - 1/2\right)^2 \right]       \\
	                & = 2 \EVp \left[ \left(s_1 - 1/2\right)^2 \right] = 2 \EVp \left[ \left(s_1 - a_0\right)^2 \right] + 2 \left(a_0 - 1/2\right)^2 = 2 a_0^2 + 2\left(a_0 - 1/2\right)^2,
\end{align*}
whose minimiser is $a_0=1/4$. This represents a trade-off between the value $a_0=0$ zeroing the variance of the only source of uncertainty in the environment which is the Gaussian draw $s_1 \sim \mathcal{N}(a_0,a_0^2)$, and the value $a_0=1/2$ splitting the total return $\mathcal{G}=1$ homogeneously into two rewards with equal expectation.

Finally, our inhomogeneous mean-volatility objective also reduces to minimisation of the risk criterion, as the expected return is 1 regardless of $\pi$. However, in this case we have no homogeneity preference, so we are free to concentrate on probabilistic risk:
\[
	\varsigma^2_\pi = \Var_\pi(s_1) + \Var_\pi(1-s_1) = 2a_0^2
\]
is obviously minimised by $a_0 = 0$.

\subsection{Deterministic-horizon example: optimal execution}\label{sec:execution}

In this subsection, we consider the common problem of an operator which needs to execute a large order on an illiquid asset, i.e., one for which the price that can be agreed for a trade worsens significantly depending on the traded size. For such assets, it is at very least too costly, if not impossible, to complete a sizeable trade without slicing it into smaller ones to be digested by the market over a long period of time. As such slicing incurs the risk of market moves during the non negligible execution time, in this situation the trader must strike a balance between acting faster with a higher negative impact on prices, and accepting higher uncertainty in the total cash-flow.

Specifically, we denote by $N_t$ the number of units of asset in the portfolio, and by $C_t$ the cash account of the trader, who starts with an exogenous initial position $N_0$ and must ensure that $N_T=0$ at a certain horizon $T$. To the purpose, at each of a discrete set of times $t_i > 0$, they can add $n_{t_i}$ (signed) units of asset registering a cash-flow of $-n_{t_i} P_{t_i}(n_{t_i})$ units of currency, where the execution price $P$ depends on the traded size $n$. Therefore, by a generic time $t$ the portfolio will contain
\[
	N_t = N_0 + \sum_{0 \leq t_i < t} n_{t_i}
\]
units of the asset, while the cash account will be
\[
	C_t = C_0 - \sum_{0 \leq t_i < t} n_{t_i} P_{t_i}(n_{t_i}).
\]
As typical in practice, we suppose that the trader computes a mark-to-market of his position using a reference price $X_t$ (often this is the mid price):
\[
	V_{t} = N_{t_i} X_{t} + C_t,
\]
and cares about the profit and loss on each period $(t_{i-1},t_i]$, so that the right reward to describe his objective in reinforcement learning terms is
\[
	R_{t_{i}} = V_{t_{i}} - V_{t_{i-1}}.
\]

To be concrete and to have a benchmark, we use the simplest and most famous model for price evolution in an optimal execution context, namely that proposed in the seminal paper \cite{almgren2000optimal}, although our reinforcement learning approach is obviously applicable to any dynamics. In particular, we suppose that the mid price $X_t$ is additively decomposed as the sum of an unaffected price $S_t$ and a correction $I_t$ depending on past trade sizes $n_{t_i}$:
\[
	X_t = S_t + I_t,
\]
where $S_t = S_0 + \sigma W_t$ is an arithmetic Brownian motion, while $I_t = \sum_{t_i \leq t} g \cdot n_{t_i}$ is a linear permanent impact term. We also follow the classical assumption that actual execution prices are worse than the reference price computed before the permanent impact by an amount which grows linearly with trade size:
\[
	P_{t_i} = X_{t_i} + \mathrm{sign}(n_{t_i})(\epsilon + \eta |n_{t_i}|).
\]

Numerically, we take the very same parameters as in section 3.4 of the original paper, i.e.~we put: $S_0 = 50\ \$$, $\sigma = 0.95\ \$\ \mathrm{day}^{-1/2}$, $g=\num{2.5e-7}\ \$$, $\epsilon = 1/16\ \$$, $\eta = \num{2.5e-6}\ \$$, with its toy low frequency $\Delta t=1\ \mathrm{day}$ and long horizon $T=5\ \mathrm{day}$ to liquidate $N_0 = \num{e6}$ shares. As for risk aversion coefficients, they will be also varied in a similar range as theirs, even though they do not have the exact same meaning because of the different risk measure; this is because we expect similar coefficients between the two approaches to lead to similar behaviour, since our reward volatility is a discrete-time analog to quadratic variation of $V_t$, and in the continuous limit the \citet{almgren2000optimal} strategy solves the mean-quadratic variation problem, as proved in \citet{forsyth2012optimal}.

Finally, we implement both the TRVO algorithm and the new IVO algorithm to optimise respectively the homogeneous and inhomogeneous mean-volatility, with a state consisting of the normalised time $t/T$, the current stock price $S_t$, and the current allocation $N_t$, although we conjecture a weak dependence of the optimal policy on the random $S_t$, as the optimiser of the classical Almgren-Chriss model is in fact a deterministic function of time only. Indeed, the optimised amounts turn out to be almost the same on all test paths. For this reason, we can simply analyse the residual inventory as a function of time on a representative test path for both algorithms.

The results of TRVO in \cref{fig:trvo} show a clear misbehaviour of homogeneous mean-volatility for this problem: immediate liquidation is never seen as the best policy, even for very high risk aversions. This is because it concentrates PnL in the first step, increasing heterogeneity of rewards and hence path volatility. This signals that this older objective function is not an adequate description of the trader's goals in this context.

This is fixed by homogeneous mean-volatility, as can be seen in \cref{fig:ivo}. Now we get a more diverse family of execution policies, where the time profile of residual inventory is more and more convex as risk aversion increases: in particular, it ranges from a linear shape for the risk-neutral trader, who splits the trade evenly across timesteps, to full execution in one step for the extremely risk averse trader. This is exactly the qualitative behaviour which makes the solutions in \cite{almgren2000optimal} financially sound.

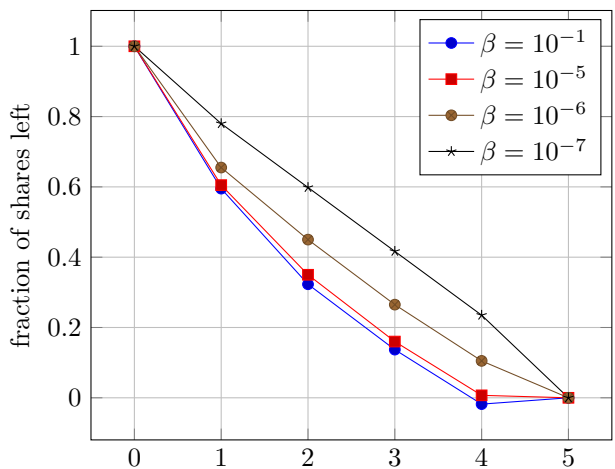
\begin{figure}
	\centering
	\begin{tikzpicture}
		\begin{axis}[
				ylabel=fraction of shares left,
				ytick scale label code/.code={},
				grid=major]
			\addplot table [col sep=comma, y=1.00E-01, x=step]{trvo.csv};
			\addlegendentry{$\beta=10^{-1}$}
			\addplot table [col sep=comma, y=1.00E-05, x=step]{trvo.csv};
			\addlegendentry{$\beta=10^{-5}$}
			\addplot table [col sep=comma, y=1.00E-06, x=step]{trvo.csv};
			\addlegendentry{$\beta=10^{-6}$}
			\addplot table [col sep=comma, y=1.00E-07, x=step]{trvo.csv};
			\addlegendentry{$\beta=10^{-7}$}
		\end{axis}
	\end{tikzpicture}
	\caption{Optimal execution paths for different risk aversion coefficients in the (homogeneous) mean-volatility objective, obtained by the TRVO algorithm.}\label{fig:trvo}
\end{figure}

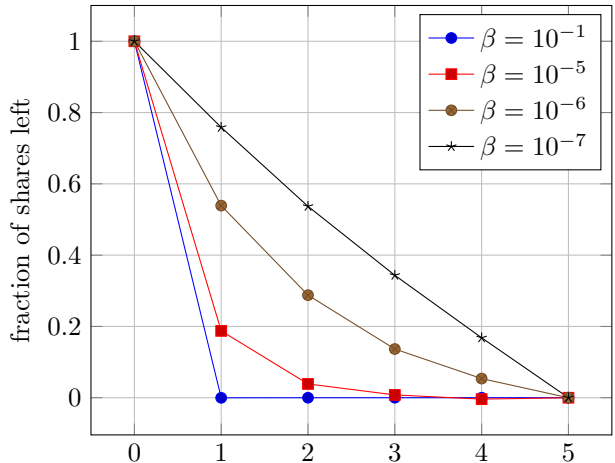
\begin{figure}
	\centering
	\begin{tikzpicture}
		\begin{axis}[
				ylabel=fraction of shares left,
				ytick scale label code/.code={},
				grid=major]
			\addplot table [col sep=comma, y=1.00E-01, x=step]{ivo.csv};
			\addlegendentry{$\beta=10^{-1}$}
			\addplot table [col sep=comma, y=1.00E-05, x=step]{ivo.csv};
			\addlegendentry{$\beta=10^{-5}$}
			\addplot table [col sep=comma, y=1.00E-06, x=step]{ivo.csv};
			\addlegendentry{$\beta=10^{-6}$}
			\addplot table [col sep=comma, y=1.00E-07, x=step]{ivo.csv};
			\addlegendentry{$\beta=10^{-7}$}
		\end{axis}
	\end{tikzpicture}
	\caption{Optimal execution paths for different risk aversion coefficients in the inhomogeneous mean-volatility objective, obtained by the IVO algorithm.}\label{fig:ivo}
\end{figure}

\subsection{Stochastic-horizon example: grid world}\label{sec:gridworld}

In this subsection, we test our objective function on an environment used in non financial reinforcement learning literature to test risk averse algorithms. This will show some non intuitive behaviour of our optimal policies when the number of timesteps per episode is stochastic.

The set of states consists of a finite rectangular grid. At each of a finite number of timesteps, the agent chooses a direction: north, south, east, west, or a diagonal direction, for a total of eight possible actions. Noise is introduced by stating that the agent moves in the selected direction only with some probability lower than 100\%, while with some residual probability a random action is executed. In both cases, if the selected or drawn action would take the agent out of the grid, the state is unaffected for the time step under consideration. The agent collects a constant small negative reward at each timestep, unless it lands on special terminal states where the episode ends; all of them carry a large negative reward except for one which can be interpreted as a goal state, where it is large and positive. We take the problem instance from \cite{moldovan2012risk}, where the grid is that of \cref{fig:grid} with -1 reward on non terminal steps, +35 reward on the terminal green cell, and -35 reward on the terminal red cells; the maximum episode length is 35 steps, the probability that a random action is picked instead of that chosen by the agent is 8\%, and the discount factor $\gamma$ equals 1.

\begin{figure}
	\centering
	\includegraphics[scale=0.5]{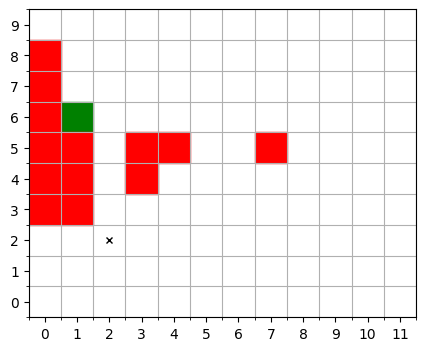}
	\caption{Grid-world instance used in \cref{sec:gridworld}, taken from \cite{moldovan2012risk}. Each episode starts from the cell marked by x and ends when the green cell is reached with reward +35, or a red cell is reached with reward -35. On every other timestep the reward is -1, for a maximum of 35 time steps.}\label{fig:grid}
\end{figure}

The intuition would be that a risk-neutral agent should take the shortest path to reach the green cell, to pay the -1 cost for as few steps as possible; while a risk averse agent should pay for longer routes to keep far from the red cells, thus reducing the probability of getting their low reward due to noise. Let us see to what extent this is confirmed under some specification of our objective function.

In particular, in the following subsections we optimise mean-$\ell$-volatility with different choices of $\ell$, each time for a grid of values of the risk aversion coefficient; unless otherwise stated, this is done by the nested approach in \eqref{eq:nested}, with the inner classical reinforcement learning problem solved exactly by tabular value iteration. Then to get an intuition on the optimised policies, we present graphically their most likely episode, i.e.~the path which is followed by the agent when noise does not disturb the action at any timestep.

\subsubsection{Inhomogeneous mean-volatility}\label{sec:meanvol}

We begin with the simplest version of our objective function, defined in \cref{def:inhomogeneous_meanvol}. The solutions we obtain for several risk aversion coefficients are in \cref{fig:meanvol}.

\begin{figure}
	\centering
	\includegraphics[width=0.9\textwidth]{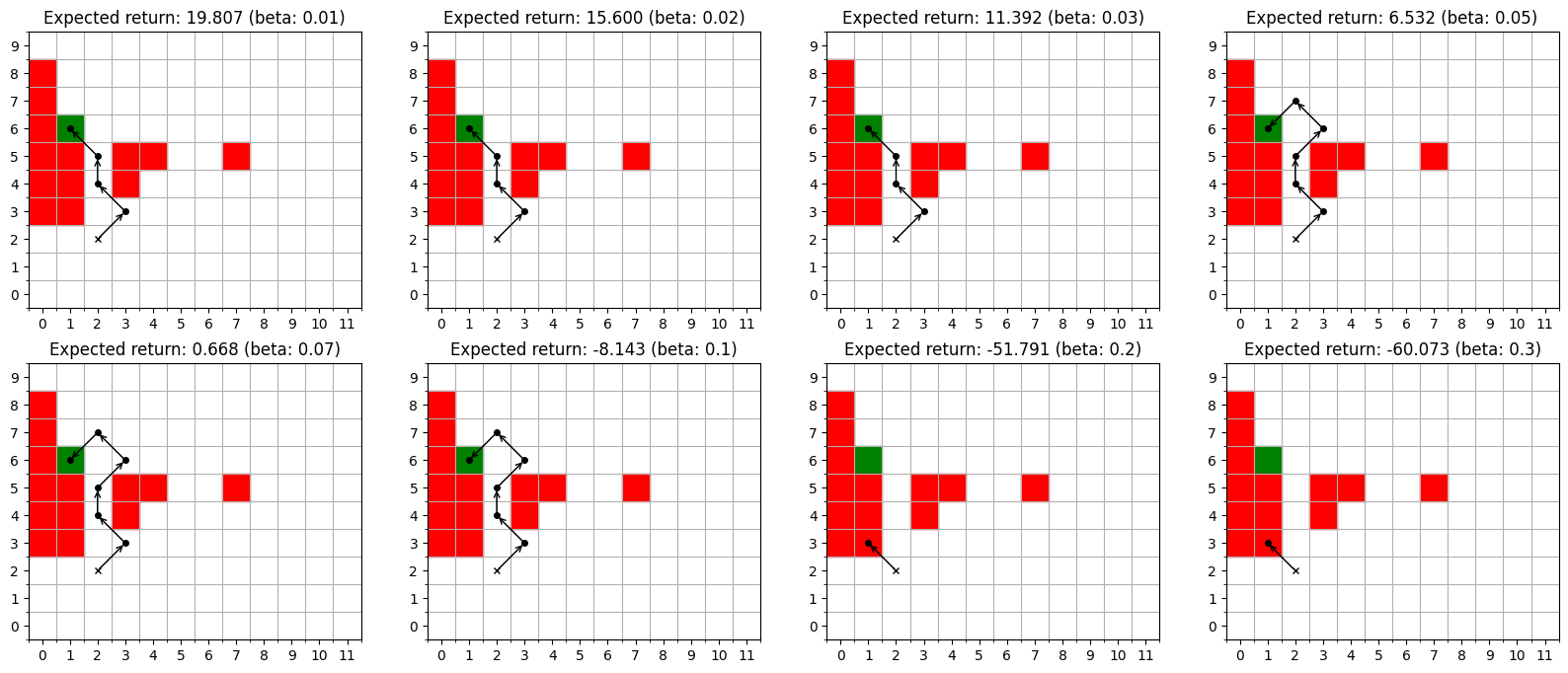}
	\caption{Path generated in a noiseless test environment by optimal policies for inhomogeneous mean-volatility \eqref{eq:inhomogeneous_meanvar}, for a selection of values of the risk aversion $\beta$.}\label{fig:meanvol}
\end{figure}

Their shape for low risk aversions is as expected: when the agent is almost risk neutral, it is optimal to take the shortest path possible. Note that the passage through cell (3,3) is strictly better than a geometrically shorter alternative path through cell (2,3), as the two paths have in fact the same length in terms of number of steps, and the former has lower probability to end in the obstacles by chance.

The shape for high risk aversion, while counter-intuitive at first sight, is also reasonable, and perhaps even more coherent than the path typically produced by other approaches, who still reach the goal cell although with a long detour. Indeed, sufficiently risk averse agents should privilege keeping as far as possible from the red squares over reaching the green one, to the point that they should aim at wandering close to the edges of the grid for the full 35 steps; but then they would get a total reward of -35 very often, and an even worse one in some unlucky episodes where noise interferes. So they are better off getting the full -35 straight away and end the episode by trying to dive into the obstacle immediately, as our agents do.

A less satisfactory behaviour is obtained for intermediate risk aversions. Indeed, the agents still take the quite risky decision of going though the narrow corridor, and the only difference from the risk neutral path is that at the fourth step of the no-noise episode, instead of choosing the action that would immediately lead to the goal state, they make a couple of waiting moves to get the +35 reward slightly later. One of the reasons for this is probably that it is suboptimal to assign a high target $\bar{r}_{\pi,4}$ to the fourth reward, as there is a non negligible number of episodes which will get delayed by noise in at least one of the first four steps, and which would contribute a strong risk penalisation. But once $\bar{r}_{\pi,4}$ is set to a low value, due to volatility equally penalizing positive and negative deviations from it, the agents will be incentivised to \emph{not} get the +35 at the fourth step even when they could as in the depicted no-noise scenario: once close enough to the goal, they will rather pay a small -1 penalty for a couple more times, and then get the +35 at the sixth step, where it is safe enough to set a high target $\bar{r}_{\pi,6}$, because the green cell can be reached with high probability in six moves even on scenarios where noise interferes.

\subsubsection{Inhomogeneous monotone mean-volatility}\label{sec:monotone}

Although the above reasoning might give a satisfactory intuitive motivation for the diamond-shaped endings in \cref{fig:meanvol}, from a human perspective they do not look very rational, so we try to change the objective function. In particular, since we partially blamed the well known flaw according to which mean-variance sometimes prefers to avoid an uncertain significant increase in reward, we try the minimal fix to this incoherence, which was defined in \cref{ex:monotone_meanvar}. The distinct types of path we get on no-noise scenarios with optimal policies corresponding to varied risk aversions are in \cref{fig:monotone}.

\begin{figure}
	\centering
	\includegraphics[width=0.9\textwidth]{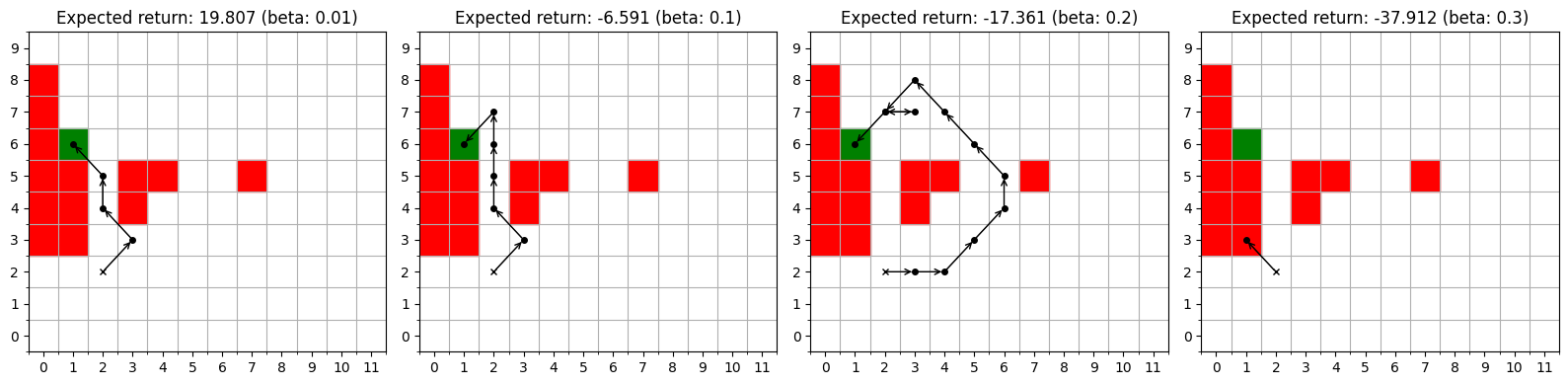}
	\caption{Path generated in a noiseless test environment by optimal policies for inhomogeneous monotone mean-volatility (\cref{ex:monotone_meanvar}), for a selection of values of the risk aversion $\beta$.}\label{fig:monotone}
\end{figure}

We can see that the good limiting cases for large and small risk aversions are retained. We also get a slightly more diverse set of optimal paths, in that for some large but not extreme risk aversions, a longer and safer path still reaching the goal state appears. But unfortunately, the main objections raised in \cref{sec:meanvol} carry over: indeed, for moderately small risk aversions, the path makes two unnecessary vertical steps before terminating, while for moderately high risk aversions, the path takes a horizontal round-trip before terminating. Therefore, removing disincentives to large rewards from earlier steps is not enough. Maybe in view of \eqref{eq:monotone}, whose interpretation is that monotone mean-variance is the mean-variance one could get having the opportunity to disregard part of the outcome, it is still better for the agent to attain the large target later but in more scenarios, rather than earlier when he must ignore part of it.

\subsubsection{Inhomogeneous optimised certainty equivalent}

A final attempt to incentivise higher returns at all steps would be to take a strictly monotone utility in the construction in \cref{ex:oce}. Results for exponential utility, i.e.
\begin{equation}\label{eq:exponential}
	{\ell}_\beta(x) = x - \frac{1 - \exp(-\beta x)}{\beta},
\end{equation}
are presented in \cref{fig:exponential}.

\begin{figure}
	\centering
	\includegraphics[width=0.9\textwidth]{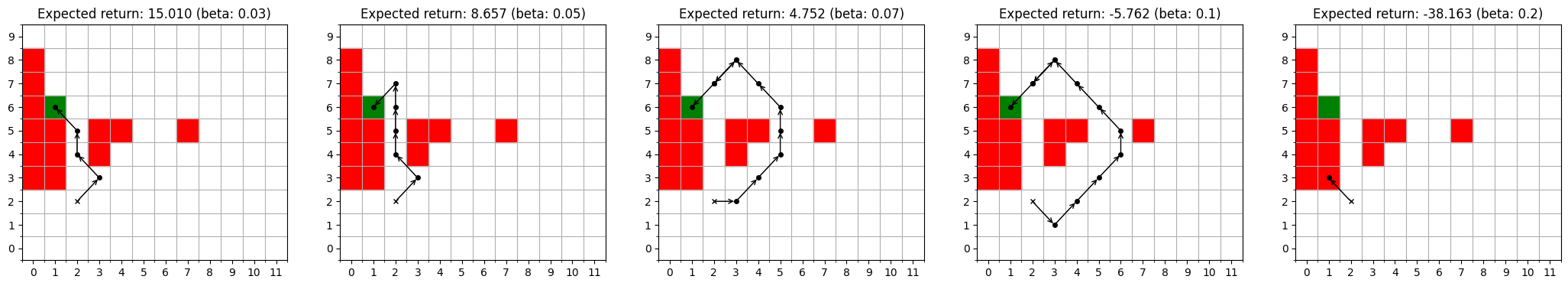}
	\caption{Path generated in a noiseless test environment by optimal policies for inhomogeneous mean-$\ell$-volatility with $\ell$ defined by \eqref{eq:exponential}, for a selection of values of $\beta$.}\label{fig:exponential}
\end{figure}

The set of paths is even richer than in \cref{sec:monotone}, and the round-trips disappear; but sadly enough, the second graph from the left still features the unnecessary upward steps we noticed in \cref{fig:monotone}.

\subsubsection{Dynamic monotone hull and IVO}

In all the tests of \cref{sec:gridworld}, the only debatable behaviours seem caused by some form of non monotonicity: the agent prefers a path which is identical except for two additional negatively-rewarded steps. For risk-measures $\rho$ in a non dynamic context, we already mentioned that a known fix is the monotone hull, which allows \emph{throwing away} part of the \emph{return} $\mathcal{G}$:
\[
	\sup_{\mathcal{Y} \in L_0^{+}} \rho(\mathcal{G}-\mathcal{Y}).
\]

This might suggest a generalisation for our setting in which the risk $\vc{\rho}$ depends on the \emph{sequence} of the rewards $r_i$: perhaps instead of simply giving up the reward, in a sequential decision setting we should
allow \emph{reserving} part of it for later steps, by optimizing
\[
	\sup_{Y_0=0, Y_{i} \in L_0^+(\mathcal{F}_i) \ \forall i>0} \vc{\rho}\left((r_i - Y_{i} + Y_{i-1})_{i\geq 1}\right).
\]

In principle, this could be implemented by just adding the real-valued reserve $Y_i$ to the control variables at time step $i$. However, with such addition, the exact solution of the inner problem becomes numerically infeasible, and one must give up the nested optimisation approach based on \eqref{eq:nested} in favour of IVO. But then, as a reinforcement learning algorithm based on function approximation (in our implementation, by feed-forward neural networks), IVO has difficulties in finding the odd slight improvements around the target state which we disliked in the above subsections. Therefore, the only tool we currently have to solve the problem with a reserve, fails to reproduce the glitches we know of even when applied without it. Hence, as of now we cannot numerically verify the conjecture that they would be fixed.

On the other hand, one can also argue more pragmatically that for this very reason, gradient based solvers may be practically effective in automatically giving reasonable solutions that smooth the edges of pathological minima even without complicating the space of actions.

\section{Conclusion}\label{sec:conclusion}

In this work, we have tried to define a risk averse objective function such that:
\begin{enumerate}
	\item Modern reinforcement learning algorithms can be leveraged without excessive overhead;
	\item The time distribution of rewards matters, unlike most proposals which define risk on the total return;
	\item Homogeneous time splits are not favoured, unlike mean-volatility which was satisfactory for hedging.
\end{enumerate}
We have proven that a time-inhomogeneous variant of mean-volatility has all the above properties, and that it can also be generalised using as building block loss functions which are better behaved than mean-variance, including monotone mean-variance and optimised certainty equivalents.

Empirically, we have seen that the proposed approach works well in the financially relevant problem of optimal execution, and we believe that the same would hold for any problem where reward timing is well predictable given the policy. On the other hand, we have observed slightly counter-intuitive optimal policies in a non financial test problem where there may be unexpected or anticipated large windfalls. While this should be quite uncommon in finance, modifications tackling these effects could be the subject of future research.

\section*{Disclaimer}

The authors report no potential competing interests. The opinions expressed in this document are solely those of the authors and do not represent in any way those of their present and past employers.

\bibliographystyle{apalike}
\bibliography{Bibl}

\end{document}